\documentclass[twocolumn,fleqn]{autart}    

\usepackage{graphicx}          

\usepackage{graphics} 
\usepackage{epstopdf}

\usepackage{amsmath} 
\usepackage{amssymb}  

\usepackage{color} 

\numberwithin{equation}{section} 

\newtheorem{theorem}{Theorem}
\newtheorem{remark}{Remark}
\newtheorem{lemma}{Lemma}
\newtheorem{proposition}{Proposition}
\newtheorem{example}{Example}
\newtheorem{corollary}{Corollary}
\newtheorem{definition}{Definition}

\newcommand{\bydef}{\stackrel{\Delta}{=}}

\newfont{\Blackboard}{msbm10 scaled 1100}
\newcommand{\bl}[1]{\mbox{\Blackboard #1}}

\newenvironment{proof}{{\it Proof :~}}{\hfill$\Box$\\}

\begin{document}

\begin{frontmatter}

\title{ \Large {\bf New Stability and Exact Observability Conditions for
 Semilinear Wave Equations}}                  

\thanks[footnoteinfo]{This work was supported by  Israel Science Foundation (grant no. 1128/14).}

\author[a]{Emilia Fridman}\ead{emilia@eng.tau.ac.il},       
\author[a]{Maria Terushkin}\ead{marinio@gmail.com}               

\address[a]{School of Electrical Engineering, Tel-Aviv University, Tel-Aviv
69978, Israel.}  


\date{}
\maketitle 

\begin{abstract}
    The problem of estimating the initial state of  1-D wave equations with globally Lipschitz nonlinearities  from boundary measurements on a finite interval
was solved recently  by using the sequence of forward and backward observers,
and deriving the upper bound for exact observability time in terms of Linear Matrix Inequalities (LMIs) \cite{fridman2013observers}.
  In the present paper, we generalize this result  to n-D  wave equations on a hypercube. This extension
   includes  new LMI-based exponential stability  conditions for n-D wave equations, as well as
 an upper bound on the minimum exact  observability time in terms of LMIs.
 For 1-D wave equations  with locally Lipschitz nonlinearities,
we  find an estimate on the region of initial conditions that are guaranteed to be
  uniquely recovered from the measurements.
The efficiency of the results is illustrated by  numerical examples.

\end{abstract}

\end{frontmatter}
\section{Introduction}

Lyapunov-based solutions of various control problems  for
finite-dimensional 
systems
 can be formulated
in the form of Linear Matrix Inequalities (LMIs) \cite{Boyd}. 
  The LMI
approach to distributed parameter systems
is capable of utilizing nonlinearities and
 of providing the desired system
{ performance}  (see e.g.
\cite{castillo2012dynamic,Aut09a,lamarelyapunov}). 
For 1-D wave equations, several control problems were solved
by using the direct Lyapunov method
in terms of LMIs \cite{Aut09b,fridman2013observers}.
However, there have not been yet LMI-based results for n-D wave equations,
though the exponential stability of the n-D wave equations in bounded
  spatial domains
   has been studied in the literature
   via the direct Lyapunov method (see e.g.
   \cite{Zuazua90,BZGuo,ammari2010,fridman2010stabilization}).

The problem of estimating the initial state of  1-D wave equations with globally Lipschitz nonlinearities  from boundary measurements on a finite interval
was solved recently  by using the sequence of forward and backward observers,
and deriving the upper bound for exact observability time in terms of LMIs \cite{fridman2013observers}.
  In the present paper, we generalize this result  to n-D  wave equations on a hypercube. This extension
   includes  new LMI-based exponential stability  conditions for n-D wave equations. 
Their derivation is based on n-D extensions of  the Wirtinger (Poincare) inequality \cite{Hardy} and of the Sobolev inequality with tight constants, which
  is crucial for
   the efficiency of the results.
   As in 1-D case,
   the continuous dependence of the reconstructed
     initial state on the measurements follows from the integral input-to-state stability
     of the corresponding error system,
      which is guaranteed by the LMIs for the exponential stability.
      Some preliminary results on global exact observability of multidimensional wave PDEs
will be presented in \cite{MashaCDC15}.

Another objective of the present paper
is to study regional exact observability for systems
with locally Lipschitz in the state nonlinearities. Here we restrict our consideration to 1-D case,
and
  find an estimate on the region of initial conditions that are guaranteed to be
  uniquely recovered from the measurements.
  Note that our result on the regional observability
 cannot be extended to multi-dimensional case (see Remark \ref{rem_reg} below for explanation and for discussion on possible n-D extensions for different classes of nonlinearities).
The efficiency of the results is illustrated by  numerical examples.

The presented simple finite-dimensional LMI conditions
complete  the theoretical qualitative results of e.g. \cite{GeorgeAut10} (where exact observability of linear systems in a Hilbert space was studied via
 a sequence of forward and backward observers)
and \cite{baroun2013semilinear} (where local exact observability of abstract semilinear systems
was considered).

{\it Notation:}
$\mathbb{R}^n$ denotes the $n$-dimensional Euclidean space with the
norm $|\cdot |$,
$\mathbb{R}^{n\times m}$ is the space of $n\times
m$ real matrices. The notation $P\!>\!0$ with
$P\in\mathbb{R}^{n\times n}$ means that $P$ is symmetric and
positive  definite. For the symmetric matrix $M$,
 $\lambda_{min}(M)$ and $\lambda_{max}(M)$
denote the
minimum and the maximum eigenvalues of $M$ respectively.
 The symmetric elements of the symmetric matrix
will be denoted by ${*}$.
Continuous functions (continuously differentiable) in all arguments,
are referred to as of class $C$ (of class $C^1$). $ L^2(\Omega)$ is
the Hilbert space of square integrable $f:\Omega\to \mathbb{R}$,
where $\Omega \subset \mathbb{R}^n$, with the norm $\|f
\|_{L^2}=\sqrt{\int_{\Omega} |f(x)|^2dx}$.
For the scalar smooth function $z=z(t,x_1,\dots, x_n)$  denote by $z_t, z_{x_k}, z_{tt}, z_{x_k x_j}$
($k,j=1, \dots n$)
the corresponding partial derivatives.
For
$z:\Omega\to \mathbb{R}$  define $\nabla z=z_x^T=[z_{x_1}\ \dots z_{x_n}]^T$,
$\Delta z=\sum_{p=1}^n z_{x_px_p}$.
${\mathcal H} ^{1}(\Omega)$ is the Sobolev space of absolutely
continuous  functions  $z:\Omega\to \mathbb{R}$ with the square
integrable $\nabla z$.
${\mathcal H} ^{2}(\Omega)$ is the Sobolev space of
  scalar functions  $z:\Omega\to \mathbb{R}$ with absolutely continuous $\nabla z$ and with $\Delta z \in L^2(\Omega)$.

\section{Observers and exponential stability of n-D wave equations}
\label{Observers}

\subsection{System under study and Luenberger type observer}
Throughout the paper we denote by $\Omega$   the n-D unit hypercube $[0,1]^n$ 
with the boundary $\Gamma$.
We use the partition of the boundary:
\begin{equation*}\begin{array}{l}
    \Gamma_D \! =\! \{x \! = \! (x_1,...,x_n)^T\in \Gamma : \
    \exists p \in {1,...,n}
    \ s.t. \ x_p =0 \}\\
   \Gamma_{N,p}= \{x  \in \Gamma: \quad x_p=1 \}, \quad \Gamma_N  =\bigcup_{p=1,\dots, n} \Gamma_{N,p}.
   \end{array}
\end{equation*}
Here subscripts D and N stand for Dirichlet and  for Neumann boundary conditions
respectively.

We consider the following boundary value problem for the scalar n-D wave equation:
\begin{equation}\label{Wave_nd_z_bound}
    \begin{array}{lll}
         z_{tt}(x,t)  &=  \Delta z(x,t)\! + \!f(z,x,t)  \; in  \; \Omega \times  (t_0,\infty ), \\
         z(x,t) &= 0 \qquad on  \qquad \Gamma_D \times  (t_0,+\infty ), \\
         \frac{\partial}{\partial \nu} z(x,t) & = 0 \qquad on  \qquad \Gamma_N \times  (t_0,\infty ),
    \end{array}
\end{equation}
where $f$ is a $C^1$ function, 
$ \nu $ denotes the outer unit normal vector to the point $ x \in \Gamma $ and $ \frac{\partial}{\partial \nu} z  $ is the normal derivative.
Let $g_1>0$ be the known bound on the derivative of $f(z,x,t)$ with respect to $z$: 
\begin{equation}
\label{g1}
  |f_z(z,x,t)|\leq g_1 \quad \forall (z,x,t)\in \bl R^{n+2}. \end{equation}
Since $\Omega$ is a unit hypercube, the boundary conditions on $\Gamma_N$ can be rewritten as
$$
z_{x_p}(x,t)\Big\rvert_{x_p=1}=0\quad \forall x_i\in[0,1], \ i\neq p, \ p=1,\dots,n.
$$
Consider the following initial conditions:
\begin{eqnarray}\label{Wave_nd_z_initial}
    z(x, t_0)= z_0(x), \ z_t(x, t_0)= z_1(x), \quad x\in \Omega. 
\end{eqnarray}
The boundary measurements are given by
\begin{equation}\label{Wave_nd_z_meas}
    y(x,t) =  z_t(x,t)  \qquad on  \qquad \Gamma_N \times  (t_0,\infty ).
\end{equation}
Similar to \cite{fridman2013observers},
the boundary-value problem (\ref{Wave_nd_z_bound}) can
be represented as an abstract differential equation by defining the
state $\zeta(t)=[\zeta_0(t)\ \zeta_1(t)]^T=[z(t) \ z_t(t)]^T$ and
the operators
$$\begin{array}{lll}
{\mathcal
A}=\left[ \begin{array}{cc} 0 & I\\
\Delta z & 0
\end{array}\right
],\quad F(\zeta,t)= \left[ \begin{array}{cc} 0 \\ F_1(\zeta_0,t)
\end{array}\right],
\end{array}$$
 where $F_1:  {\mathcal
H}^1(\Omega)\times R \to L^{2}(\Omega) $  is  defined as
$F_1(\zeta_0,t)=f(\zeta_{0}(x), x,t)$ so that it is continuous in
$t$ for each $\zeta_0\in{\mathcal H}^1(\Omega)$. The differential equation
is
\begin{equation}
\label{x} \dot \zeta(t)={\mathcal A} \zeta(t)+F(\zeta(t),t), \quad
t\geq t_0
\end{equation}
 in the Hilbert space ${\mathcal
H}={\mathcal H}_{\Gamma_D}^{1}(\Omega)\times L^2(\Omega)$, where
$${\mathcal H}^{1}_{\Gamma_D}(\Omega) =\left\{\zeta_0\in {\mathcal H}^{1}(\Omega)\ \bigg| \ {\zeta_0}_{|\Gamma_D}=0\right \}$$
and $\|\zeta\|^2_{\mathcal
H}=\|\nabla\zeta_{0}\|^2_{L^2}+\|\zeta_{1}\|^2_{L^2}$.
The operator
 ${\mathcal
A}$ has the dense  domain
$$ 
\begin{array}{lll} {\mathcal
D}({\mathcal A})=
\!\!\Big\{&\!\!(\zeta_0,\zeta_1)^T\in {\mathcal
H}_{\Gamma_D}^{1}(\Omega)\times {\mathcal H}_{\Gamma_D}^{1}(\Omega)
  \bigg|\ \Delta \zeta_0\in L^2(\Omega) \\ & {\rm and}\
  \frac{\partial}{\partial \nu}{\zeta_0}_{|\Gamma_N}\!=-b{\zeta_1}_{|\Gamma_N}\! \Big\},
\label{A20}\end{array}$$
where $b=0$.
Here the boundary condition holds in a weak sense (as defined in Sect. 3.9
of \cite{GeorgeBook}), i.e. the following relation holds:
$$\begin{array}{r}
\langle \Delta \zeta_0,\phi\rangle_{L^2(\Omega)}+
\langle \nabla \zeta_0,\nabla \phi\rangle_{[L^2(\Omega)]^n}=-b
\langle \zeta_0,\phi\rangle_{L^2(\Gamma_N)} \\ \forall \phi \in {\mathcal H}_{\Gamma_D}^{1}(\Omega).
\end{array}$$

  The operator ${\mathcal A}$ is m-dissipative
 (see Proposition 3.9.2 of \cite{GeorgeBook}) and hence it generates a strongly  continuous
 semigroup.
Due to (\ref{g1}),
the following Lipschitz condition holds: 
\begin{equation}
\begin{array}{lll}
\label{LipW} \|F_1(\zeta_{0}, t)-F_1( \bar\zeta_{0}, t)\|_{L^2} \leq
g_1\|\zeta_{0}-\bar\zeta_{0}\|_{L^2}
\end{array}
\end{equation}
where $\zeta_{0},\bar\zeta_{0} \in {\mathcal H}^{1}_{\Gamma_D}(\Omega), t\in \bl R.$ Then by Theorem 6.1.2 of \cite{pazy1983semigroups}, a
unique continuous mild solution $\zeta(\cdot)$ of (\ref{x}) in
${\mathcal H}$ initialized by
\begin{equation*}\label{wave_in1}
\begin{array}{ll}
\zeta_0(t_0)=z_0\in 
  {\mathcal H}^{1}_{\Gamma_D}(\Omega),  \
 \zeta_1(t_0)=z_1\in
L^2(\Omega)
\end{array}
\end{equation*}
   exists in  $C([t_0, \infty), {\mathcal H})$.
  If $\zeta(t_0) \in  {\mathcal
D}({\mathcal A})$, then this mild solution is in $C^1([t_0, \infty),
{\mathcal H})$
 and it is a classical solution of (\ref{Wave_nd_z_bound}) with $\zeta(t) \in  {\mathcal
D}({\mathcal A})$ (see Theorem 6.1.5 of \cite{pazy1983semigroups}).

We suggest a Luenberger type observer of the form:
\begin{equation}\label{observer_nd}
   \widehat{z}_{tt}(x,t)  = \Delta \widehat{z}(x,t)  + f\Big(\widehat{z} ,x,t \Big),
   \qquad t \geq t_0, \; x \in \Omega
\end{equation}
under the  initial conditions
  $[{\widehat z}(\cdot,t_0), {\widehat
z}_t(\cdot,t_0)]^T\in {\mathcal H}$ and the
boundary conditions
\begin{equation}\label{observer_nd_boundary}
\begin{array}{llll}
    \widehat{z}(x,t)&=0                                                                     &  on \  \Gamma_D\times(t_0,\infty)\\
    \frac{\partial}{\partial \nu} \widehat{z}(x,t)& = k\Big[ y(x,t)- \widehat{z}_t(x,t)  \Big] \!& on \   \Gamma_N\times(t_0,\infty)
    \end{array}
\end{equation}
where $ k $ is the injection gain.

 The well-posedness of (\ref{observer_nd}), (\ref{observer_nd_boundary}) will be
established by showing the well-posedness of the estimation error
 $ e = z - \widehat z $.
Taking into account \eqref{Wave_nd_z_bound}, \eqref{Wave_nd_z_initial} we obtain
the following PDE for the estimation error $ e=z-\widehat{z} $:
\begin{equation}\label{Wave_nd_e}
     e_{tt}(x,t) =\Delta e(x,t)  + g  e(x,t)   \quad t\geq t_0, \quad x \in \Omega
\end{equation}
under the boundary conditions
\begin{equation}
    \begin{array}{lll}\label{bound_e}
         e(x,t)&=0                                             &  on   \  \Gamma_D\times(t_0,\infty)\\
             \frac{\partial}{\partial \nu} e(x,t)& = -k e_t(x,t)   &  on \ \Gamma_N\times(t_0,\infty).
    \end{array}
\end{equation}
Here $g e=f (z,x,t)- f(z-e,x,t)$  and
\begin{equation*}
    \begin{array}{lll}
    g=g( z,e,x,t) =\int_0^1 f_z( z +(\theta-1)e, x,t) d\theta.
    \end{array}
\end{equation*}
The initial conditions for the error are given by
\begin{equation*}
    \begin{array}{lll}
          e(x,t_0) &= z_1(x)-z(\cdot,t_0),    \\
          e_t(x,t_0) &= z_2(x)- z_t(\cdot,t_0)
    \end{array}
\end{equation*}
The boundary conditions on $\Gamma_N$ can be presented as
\begin{equation}
    \begin{array}{lll}\label{Wave_e_Neum}
    e_{x_p}(x,t)\Big\rvert_{x_p=1}=-k e_t(x,t)\quad \forall x_i\in[0,1],\\
    \nonumber \ i\neq p, \ p=1,\dots,n.
    \end{array}
\end{equation}

Let $z$ be a mild solution of \eqref{Wave_nd_z_bound}. 
Then $z:[t_0,\infty) \to {\mathcal H}^1$ is continuous
and, thus, the function $F_2:  {\mathcal H}^1\times [t_0,\infty)\to
L_{2}(0,1) $    defined as
$$F_2(\zeta_0,t)=
f (z,x,t)- f(z-\zeta_{0},x,t)
$$
satisfies the Lipschitz condition (\ref{LipW}), where $F_1$ is
replaced by $F_2$. By the above arguments, where in the definition of ${\mathcal
D}({\mathcal A})$ we have $b=k$, the error system
(\ref{Wave_nd_e}), \eqref{bound_e} has a unique mild solution $\{e,
e_t\} \in C([t_0, \infty), {\mathcal H})$ initialized by
$[e(\cdot,t_0), e_t(\cdot,t_0)]^T \in {\mathcal H}.$ Therefore, there
exists a  unique  mild solution $\{\hat z, \hat z_t\} \in C([t_0,
\infty), {\mathcal H})$ to the observer system
\eqref{observer_nd}, \eqref{observer_nd_boundary} with the  initial conditions $[\hat
z(\cdot,t_0), \hat z_t(\cdot,t_0)]^T\in {\mathcal H}$. If
$[e(\cdot,t_0), e_t(\cdot,t_0)]^T \in  {\mathcal D}({\mathcal A}),$
then
 $\{e, e_t\}
\in C^1([t_0, \infty), {\mathcal H})$ is a classical solution of
\ref{Wave_nd_e}), \eqref{bound_e} with $[e(\cdot,t), e_t( \cdot,t)]\in
{\mathcal D}({\mathcal A})$  for $t\geq t_0$.
 Hence, if  $[\hat z(\cdot,t_0), \hat z_t(\cdot,t_0)]^T\in {\mathcal
D}({\mathcal A})$ and $[z_0, z_1]^T\in {\mathcal D}({\mathcal A})$,
there exists a  unique  classical solution $\{\hat z, \hat z_t\} \in
C^1([t_0, \infty), {\mathcal H})$ to the observer system
\eqref{observer_nd}, \eqref{observer_nd_boundary} with
 $[\hat z(\cdot,t), \hat z_t(\cdot,t)]^T\in {\mathcal
D}({\mathcal A})$ for $t\geq t_0$.

\subsection{Lyapunov function and useful inequalities}

We will derive further sufficient conditions for the exponential
stability of the error wave equation \eqref{Wave_nd_e} under the
boundary conditions \eqref{bound_e}.
 Let
\begin{equation}\label{Ee}
    \begin{array}{ll}
    E(t) =  \frac{1}{2} \int_{\Omega}
    \left[|\nabla e|^2  +  e_t^2 \right]    dx,
    \end{array}
\end{equation}
be the energy of the system.
Consider the following Lyapunov function for  \eqref{Wave_nd_e}, \eqref{bound_e}:
\begin{equation}\label{Lyapunov_nd}
    \begin{array}{ll}
      V(t)
    =E(t) + \chi  \int_{\Omega}
    \left[ 2  (x^T \cdot \nabla e ) +  (n-1)e \right] e_t  dx \nonumber\\
     \quad +\chi \frac{k(n-1)}{2}\int_{\Gamma_N}e^2d\Gamma
  \end{array}
\end{equation}
with some constant $ \chi > 0 $. Note that the above
Lyapunov function
without the last term was considered in \cite{ammari2010,fridman2010stabilization,Zuazua90}.
The time derivative of this new term of $V$ cancels the same term with the opposite sign in the time derivative of
$\chi  \int_{\Omega}[(n-1)e] e_t  dx$ (cf. \eqref{identitymul} below) leading to LMI conditions
for the exponential convergence of the error wave equation.

We will employ the following n-D extensions of the classical
inequalities:
\begin{lemma}\label{lem_ineq}
    Consider  $e\in {\mathcal H}^{1}(\Omega)$ such that
    $e\Big\rvert_{x\in \Gamma_D}=0$.
    Then the following n-D Wirtinger's inequality holds:
\begin{equation} \label{Wir}
    \begin{array}{ll}
  \int_{\Omega} \left[ \frac{4}{\pi^2 n} |\nabla e|^2 -e^2 \right] dx \geq 0.
  \end{array}
\end{equation}
    Moreover,
\begin{equation} \label{lambda2}
    \begin{array}{ll}
   \int_{\Gamma_N} e^2  d \Gamma \leq \int_{\Omega} |\nabla e|^2 dx.
      \end{array}
\end{equation}
\end{lemma}

\begin{proof}
Since $e\Big\rvert_{x_1=0}=0$, by the classical 1-D Wirtinger's inequality \cite{Hardy}
\begin{equation*}
    \begin{array}{ll}
    \int_0^1 e^2 dx_1\leq \frac{4}{\pi^2 } \int_0^1 e_{x_1}^2 dx_1.
    \end{array}
\end{equation*}
Integrating the latter inequality in $x_2,\dots, x_n$ we obtain
\begin{equation*}
    \begin{array}{ll}
    \int_{\Omega} e^2 dx\leq \frac{4}{\pi^2 } \int_{\Omega} e_{x_p}^2 dx
    \end{array}
\end{equation*}
with $p=1$. Clearly the latter inequality holds for all $p=1,\dots, n$, which after
summation in $p$  yields \eqref{Wir}.

Since $e\Big\rvert_{x_1=0}=0$ we have by Sobolev's inequality
\begin{equation*} \label{ex1bound}
    \begin{array}{ll}
    e^2(x)\Big\rvert_{x_1=1}\leq \int_0^1 e_{x_1}^2 dx_1 \ \forall x_i\in[0,1],\ i\neq 1,
    \end{array}
\end{equation*}
that after integration in $x_2,\dots,x_n$ leads to
\begin{equation*}
    \begin{array}{ll}
    \int_{\Gamma_{N,p}} e^2  d \Gamma \leq \int_{\Omega} e_{x_p}^2 dx
    \end{array}
\end{equation*}
with $p=1$. The latter inequality holds  $\forall p=1,\dots,n$ leading after summation in $p$ to \eqref{lambda2}.
\end{proof}


\subsection{Exponential stability of n-D 
wave equation}

In this section we derive LMI conditions for the exponential stability of the estimation error equation.
We start with the conditions for the positivity of the Lyapunov function:

\begin{lemma}\label{alpha_beta_section}
Let there exist positive scalars $\chi$ and $\lambda_0$ such  that
\begin{equation}\label{Phi0}
    \Phi_0\bydef
       \begin{bmatrix}
       \frac{1}{ 2}- \lambda_0 \frac{4}{\pi^2 n}& \sqrt{n}\chi &
       0 \\
            {*}       & \frac{1}{ 2}& \frac{n-1}{ 2}\chi         \\
            {*}       & {*}        & \lambda_{0}
        \end{bmatrix}>0.
\end{equation}
       Then the
    Lyapunov function $ V(t) $ is bounded as follows:
    \begin{equation}
        \begin{array}{ll}\label{Vpos}
            \alpha E(t) \leq V(t) \leq \beta E(t),\quad
            \alpha = 2\lambda_{min} (\Phi_0), \\
            \beta  = 2\left(1+\frac{2}{\pi^2n}\right)\lambda_{max} (\Phi_1)+\chi k(n-1),
        \end{array}
    \end{equation}
    where $\Phi_1=\Phi_0+diag\{\lambda_0 \frac{4}{\pi^2 n}, 0, 0\}$.
\end{lemma}
\begin{proof}
By Cauchy-Schwarz inequality
we have
\begin{equation}\label{CSchw}
    \begin{array}{lll}
    |x^T \cdot \nabla e|
    \leq
    |x| |\nabla e|
    \leq
    \sqrt{n}|\nabla e|,
    \end{array}
\end{equation}
Then
\begin{equation*}
    \begin{array}{l}
        |\chi  \int_{\Omega}
        \left[ 2  (x^T \cdot \nabla e ) +  (n-1)e \right] e_t  dx|\\
        \leq
            \chi \int_{\Omega}[2\sqrt{n}|\nabla e|| e_t |+(n-1)|e| |e_t |] dx,
    \end{array}
\end{equation*}
leading to
\begin{equation}
    \begin{array}{ll}\label{Vbound1}
        V(t)  \geq  \frac{1}{ 2}\int_{\Omega} \left[ e_t^2+\vert\nabla e\vert^2 \right]dx \\
             \quad    -\chi \int_{\Omega} [ 2\sqrt{n}|\nabla e|| e_t | +(n-1)|e| |e_t | ]dx.
    \end{array}
\end{equation}
Taking into account the n-D Wirtinger inequality \eqref{Wir}, we further
 apply S-procedure \cite{Yak} {\footnote {Let $A_i\in {\mathbb R}^{p\times p}, \ i=0,1.$
 Then the inequality
 $\xi^TA_0\xi\geq 0$ holds  for any $\xi\in {\mathbb R}^p$ satisfying $\xi ^T
A_1\xi\geq 0 $
iff there exists a real scalar $\lambda\geq 0,$ such that
$
 A_0-\lambda A_1\geq 0$.}}, where we subtract
 from the right-hand side of \eqref{Vbound1} the nonnegative term
\begin{equation}\label{spr}
    \lambda_0 \int_{\Omega} \left[ \frac{4}{\pi^2 n} |\nabla e|^2 -e^2 \right] dx
\end{equation}
 with $\lambda_0>0$:
\begin{equation*}
    \begin{array}{l}
     V(t)\geq \frac{1}{ 2} \int_{\Omega}\{ e_t^2+\vert\nabla
    e\vert^2\}dx\\
    -\chi \int_{\Omega} [ 2\sqrt{n}|\nabla e|| e_t | +(n-1)|e| |e_t | ]dx\\
    -\lambda_0\int_{\Omega}\left[ \frac{4}{\pi^2 n} |\nabla e|^2 -e^2 \right]dx
    = \int_{\Omega}\eta^T\Phi_0\eta,
    \end{array}
\end{equation*}
where $\eta=col \{|\nabla e|, -|e_{t}|, |e|\}$.

Similarly
\begin{eqnarray}\label{Vbound2}
     V(t)     & \leq &  \frac{1}{2} \int_{\Omega}\left[ e_t^2+\vert\nabla e\vert^2 \right]dx \\
    \nonumber & +    & \chi \int_{\Omega} \left[ 2\sqrt{n}|\nabla e|| e_t | +(n-1)|e| |e_t | \right]dx \\
    \nonumber & +    &\chi\frac{k(n-1)}{ 2}\int_{\Gamma_N}e^2d\Gamma \\
    \nonumber & \leq & \eta_1^T\Phi_1\eta_1+\chi \frac{k(n-1)}{ 2}\int_{\Gamma_N}e^2d\Gamma
\end{eqnarray}
with $\eta_1= col\{|\nabla e|, |e_{t}|, |e|\}$.

Then \eqref{Vpos} follows from
\begin{equation*}
    \begin{array}{l}
    \lambda_{min}(\Phi_0) [2E(t)+\int_{\Omega}e^2dx]\leq V(t) \\
    \leq \lambda_{max}(\Phi_1) [2E(t)+\int_{\Omega}e^2dx]+\chi \frac{k(n-1)}{ 2}\int_{\Gamma_N}e^2d\Gamma
    \end{array}
\end{equation*}
and from the  inequalities \eqref{Wir} and \eqref{lambda2}.
%
%
%
\end{proof}

We are looking next for conditions that  guarantee
$\dot V(t)+2\delta V(t) \leq 0$
 along the classical solutions of the wave equation initiated from
$[z_0, z_1]^T, [\hat z (\cdot, t_0), \hat
z_t(\cdot, t_0)]^T \in{\mathcal D}({\mathcal A})$. Then
$V(t)\leq e^{-2\delta (t-t_0)}V(t_0)$ and, thus,    \eqref{Vpos}
yields
\begin{flalign}\label{lem_enc}
    &   \int_{\Omega} [| \nabla e|^{2}(x,t)+e^{2}_t(x,t)]dx \\
    &   \quad \leq
        \frac{\beta}{ \alpha} e^{-2\delta (t-t_0)} \int_{\Omega}[| \nabla (z_{0}(x) - \hat z(x, t_0))|^2 \nonumber \\
    & \quad     +  (z_1(x)-\hat z_{t}(x,t_0))^2]dx. \nonumber
\end{flalign}
 Since ${\mathcal D}({\mathcal A})$ is dense in ${\mathcal H}$ the same estimate \eqref{lem_enc}
remains true (by continuous extension) for any initial conditions
$[z_0, z_1]^T, [\hat z (\cdot, t_0), \hat z_t(\cdot, t_0)]^T
\in{\mathcal H}$. For such initial conditions we have mild solutions
of \eqref{Wave_nd_z_bound}, \eqref{Wave_nd_z_initial}.
%
\begin{theorem}
    \label{wave_nd_exp}
    Given $k>0$ and $\delta >0$, assume that there exist
    positive constants $\chi, \lambda_0$ and $\lambda_1$  that satisfy
    the LMI
      \eqref{Phi0} 
     and  the following LMIs:
\begin{equation} \label{Psi1_4}
    \begin{array}{l}
        \Psi_1 \bydef - k +  (1+ k^2n)\chi
             \leq 0, \\
         \Psi_2 \bydef
            \begin{bmatrix}
           \psi_{2}& 2\delta \sqrt{n}\chi & \sqrt{n} g_1\chi\\
            * & -\chi+\delta & \frac{1}{ 2}g_1+\delta(n-1)\chi\\
            * & * & -\lambda_ 1 +g_1  (n-1)\chi
            \end{bmatrix}
            \leq 0,\\
    \psi_{2}=-\chi+\delta(1+\chi k(n-1))+ \lambda_1 \frac{4}{\pi^2 n}.
\end{array}
\end{equation}
Then,  under the condition \eqref{g1},
solutions of the boundary-value problem \eqref{Wave_nd_e},
    \eqref{bound_e} satisfy \eqref{lem_enc}, where $\alpha$ and
    $\beta$ are given by \eqref{Vpos}, i.e. the system governed by
    \eqref{Wave_nd_e}, \eqref{bound_e} is exponentially stable with a
    decay rate $\delta>0$.
\end{theorem}

\begin{proof}
Differentiating $V$ in time 
we obtain
\begin{equation*}
    \begin{array}{l}
    \dot V(t) = \dot E(t)+\chi \frac{d}{ dt} \Big[ \int_{\Omega}
    \left[ 2  (x^T \cdot \nabla e ) +  (n-1)e \right] e_t  dx\Big] \\
    \quad +\chi k(n-1)\int_{\Gamma_N}e e_t d\Gamma
    \end{array}
\end{equation*}
We have
\begin{equation*}
    \begin{array}{lll}
    \dot{E}(t)= \int_\Omega \left( (\nabla e)^T \nabla (e_t) + e_t e_{tt} \right) dx .
    \end{array}
\end{equation*}
Applying Green's formula to the first integral term, substituting $e_{tt}=\Delta e+g e$ and taking into account \eqref{g1}, we find
\begin{equation*}
    \begin{array}{lll}
          \dot{E}(t) & = \int_\Gamma e_t\frac{\partial e}{\partial \nu} d \Gamma
                                  - \int_\Omega e_t \Delta e  dx + \int_\Omega e_t[\Delta e+ g e] dx \\
                      & \leq -k \int_{\Gamma_N} e^2_t d \Gamma + g_1 \int_\Omega |e| |e_t| dx
    \end{array}
\end{equation*}
 Furthermore, we have
\begin{equation*}
    \begin{array}{l}
       \frac{d}{dt} \left \{ \int_{\Omega} [2x^T\nabla e +(n-1) e] e_t  dx\right\} \\
       \quad =   \int_{\Omega} \frac{d}{dt}[2x^T\nabla e
                    +(n-1) e] e_t  dx \nonumber \\
       \quad     +  \int_{\Omega} [2x^T\nabla e     +(n-1) e] [\Delta e +ge]  dx. \nonumber
    \end{array}
\end{equation*}
Then Green's formula leads to (see (11.35) of \cite{lions1988exact})
\begin{equation}\label{nVterm}
    \begin{array}{l}
     \frac{d}{dt} \left \{ \int_{\Omega} [x^T\nabla e +(n-1) e ] e_t  dx\right\} \\
     \quad =2 \int_{\Gamma_N} x^T\nabla e\ \frac{\partial e}{\partial\nu}d\Gamma- \int_{\Gamma_N}(x^T\nu)\vert\nabla e\vert^2 d\Gamma \\
    \quad  +(n-2) \! \int_{\Omega}\vert\nabla e\vert^2 dx +\int_{\Gamma_N}(x^T\nu)e_t^2d\Gamma
    -n \! \int_{\Omega} e_t^2 dx\\
     \quad  +(n-1) \! \int_{\Omega}
    e_t^2 dx+(n-1)\int_{\Gamma_N} e\frac{\partial
    e}{\partial\nu} d\Gamma\\
     \quad -(n-1)\! \int_{\Omega}\vert\nabla e\vert^2 dx
    +  \int_{\Omega} [2x^T\nabla e
    +(n-1) e]ge  dx
    \end{array}
\end{equation}
Noting that
$ x^T\nu=1$  on $\Gamma_N$
and taking into account the boundary conditions we obtain
\begin{equation}
\begin{array}{l}\label{identitymul}
    \frac{d}{dt} \left \{ \int_{\Omega} [2x^T\nabla e
    +(n-1) e] e_t  dx\right\}  \\
\quad  =-\int_{\Omega}\{ e_t^2+\vert\nabla
    e\vert^2 + [2x^T\nabla e
    +(n-1) e] g e \} dx
    \\
    \quad -  \int_{\Gamma_N} \! \left[\vert\nabla e\vert^2 \! + \! 2 kx^T\nabla ee_t\right ] d\Gamma
    \\
    \quad +  \int_{\Gamma_N} \! \left[e_t^2 \! - \! k(n-1)e e_t\right]d\Gamma
    \end{array}
\end{equation}

By 
inequalities \eqref{CSchw} and
 \eqref{g1} we have
\begin{equation*}
    \begin{array}{l}
    \int_{\Omega}[2x^T\nabla e
    +(n-1) e] g e dx \leq \int_{\Omega}\Big[2|x^T\nabla e| |g| |e| dx \\
     +(n-1)g_1 e^2 \Big]dx \leq
    \int_{\Omega}[2\sqrt{n} g_1 |\nabla e|  |e|+(n-1)g_1 e^2]dx.
    \end{array}
\end{equation*}
Further due to \eqref{CSchw}
\begin{equation*}
    \begin{array}{l}
           - \int_{\Gamma_N}  2 kx^T\nabla ee_t d\Gamma  \\
 \quad     \leq  2k  \int_{\Gamma_N}  |x^T\nabla e||e_t| d\Gamma
            \leq  2k \sqrt{n}   \int_{\Gamma_N}  |\nabla e||e_t|d\Gamma.
    \end{array}
\end{equation*}
Then by completion of squares we find
\begin{equation*}
    \begin{array}{l}
     -\int_{\Gamma_N}\left[\vert\nabla e\vert^2 +2 kx^T\nabla ee_t\right ] d\Gamma \\
     \; \leq \! \int_{\Gamma_N} \!\! \Big[k^2ne_t^2 -\left[\vert\nabla e\vert \! - \! k\sqrt{n} |e_t|\right]^2 \Big]d\Gamma \!
     \leq \! k^2n \! \int_{\Gamma_N} \! e_t^2 d\Gamma.
    \end{array}
\end{equation*}

Summarizing we obtain
\begin{equation}
    \begin{array}{l}\label{dotV}
      \dot V(t) \leq [\chi(1+k^2n)-k] \int_{\Gamma_N}\!  e_t^2 d \Gamma
                \!-\!\int_{\Omega}\Big [\chi[ e_t^2+\vert\nabla     e\vert^2] \\
             \! -  \! \left[2\sqrt{n}\chi g_1 |\nabla e|  |e| \! + \! (n \! - \! 1)\chi g_1 e^2 \! + \! g_1|e| |e_t|\right] \!\Big] dx.
    \end{array}
\end{equation}
Therefore, employing \eqref{Vbound2} we arrive at
\begin{equation}
    \begin{array}{l}\label{dotVV}
     \dot V(t)  +   2\delta V(t)      \leq      \int_{\Gamma_N}[\Psi_1  e_t^2+ \delta\chi{k(n-1)}e^2] d \Gamma \\
     \quad - (\chi-\delta)\int_{\Omega}\left[ e_t^2+\vert\nabla e\vert^2\right] dx \\
     \quad +  \int_{\Omega}\Big[2\sqrt{n}\chi g_1 |\nabla e|  |e| + (n-1)\chi g_1 e^2+\\
     \quad + [g_1+ 2\delta\chi(n-1)]|e| |e_t|+ 4\delta \chi \sqrt{n}|\nabla e|| e_t |\Big]dx.
    \end{array}
\end{equation}
By taking into account Wirtinger's inequality \eqref{Wir}, we add to \eqref{dotVV} the nonnegative term \eqref{spr}, where
$\lambda_0$ is replaced by $\lambda_1>0$.
Denote
$\eta_2 = col \{|\nabla e|, |e_{t}|, |e|\}$.
Then after employing the bound \eqref{lambda2} we arrive at
\begin{equation*}
    \begin{array}{l}
        \frac{d}{dt} V(t) + 2 \delta V(t)
        \leq
        \Psi_1 \int_{\Gamma_N}  e_t^2 d \Gamma
        +  \int_\Omega  \eta_2^T \Psi_2 \eta_2 dx
        \leq 0
    \end{array}
\end{equation*}
if the LMIs \eqref{Psi1_4} are feasible.
\end{proof}

\begin{remark}
    For $n>1$ the term $\chi  \int_{\Omega}
    (n-1)e e_t  dx$  of $V$
    leads  to $-\chi\int_{\Omega}\vert\nabla e\vert^2 dx$ in $\dot V$ (cf. \eqref{nVterm}).
\end{remark}

\section{Exact observability of n-D wave equation}\label{sec_ExObs}

Our next objective is to recover (if possible) the unique initial
state \eqref{Wave_nd_z_initial}  of the solution to
\eqref{Wave_nd_z_bound}-\eqref{Wave_nd_z_initial} from the measurements on the finite
time interval
\begin{equation}
    \begin{array}{ll}\label{measW}
        y(x,t)=z_t(x,t)\ on  \ \Gamma_N \times [t_0, t_0+T], \ T>0.
    \end{array}
\end{equation}

\begin{definition} \cite{fridman2013observers}
 The system \eqref{Wave_nd_z_bound}, \eqref{Wave_nd_z_initial} 
with the measurements \eqref{measW} is called exactly observable in
time $T$, if

(i) for any initial condition $[z_0, z_1]^T\in {\mathcal
H}={\mathcal H}_{\Gamma_D}^{1}(\Omega)\times L^2(\Omega)$
  it is possible to find a sequence  $[z_0^m, z_1^m]^T\in {\mathcal H} (m=1,2,...)$ from the measurements
  \eqref{measW}
such that
$\lim_{m\to
\infty}\|[z_0^m, z_1^m]^T-[z_0, z_1]^T\|_{\mathcal H}=0$
(i.e. it is possible to recover the unique initial state as
$[z_0, z_1]^T=\lim_{m\to\infty}[z_0^m, z_1^m]^T$);

(ii) there exists a constant $C>0$ such that
for any initial conditions $[z_0, z_1]^T\in {\mathcal H}$ and  $[\bar z_0, \bar z_1]^T\in {\mathcal H}$ leading to the measurements $y(x,t)$ and $
\bar y(x,t)$ and to the corresponding sequences $[z_0^m, z_1^m]^T$ and
$[\bar z_0^m, \bar z_1^m]^T$, the following holds: 
\begin{equation}
    \begin{array}{l}\label{def_exact}
    \|\lim_{m\to \infty}[z_0^m, z_1^m]^T-\lim_{m\to \infty}[\bar z_0^m, \bar z_1^m]^T\|^2_{\mathcal{H}} 
   \\
    \leq
    C\int_{t_0}^{t_0+T}\int_{\Gamma_N}|y(x,t)-\bar y(x,t)|^2d\Gamma dt.
    \end{array}
\end{equation}
The time $T$ is called the observability time.

The system is called {\it regionally exactly observable} if the above conditions hold for
all $[z_0, z_1]^T\in {\mathcal H}$ with $\|[z_0, z_1]^T\|_{\mathcal{H}}\leq d_0$ for some $d_0>0$.
\end{definition}

Note that \eqref{def_exact} means the continuous in the measurements recovery of the initial state.
In this section we will derive LMI sufficient conditions for n-D wave equations
with globally Lipschitz in the first argument $f$, where
\eqref{g1} holds globally in $z$.
In Section \ref{regional_obs_section}, we will present
  LMI-based conditions for the
 regional observability
 for 1-D wave equation, where \eqref{g1} holds locally
in $z$.

\subsection{Iterative forward and backward observer design}
In order to recover the initial state
  of the solution to
\eqref{Wave_nd_z_bound}
from the measurements \eqref{measW}
 we use the iterative procedure
as in \cite{GeorgeAut10}. Define the sequences of forward $z^{(m)}$ and
backward observers $z^{b(m)}, \ m=1,2,...$ with the injection gain
$k>0$:
\begin{equation}\label{zn}
    \begin{array}{lllll}
       z^{(m)}_{tt}(x,t)=
         \Delta z^{(m)}(x,t)+f(z^{(m)}(x,t),  x,t), \\
         z^{(m)}(x,t)  =
          0,\quad x \in \Gamma_D, \\
         \frac{\partial}{\partial \nu} z^{(m)}(x,t)  = k [ y(x,t)- z_t^{(m)}(x,t)],\quad x \in \Gamma_N,   \\
          \qquad t\in[t_0,t_0+T],\\
         z^{(m)}(x,t_0)=z^{b(m-1)}(x,t_0), \\
          z^{(m)}_t(x,t_0)= z_t^{b(m-1)}(x,t_0),
    \end{array}
\end{equation}
where
$z^{b(0)}(x,t_0)= z_t^{b(0)}(x,t_0)\equiv 0$,
 and
\begin{equation}\label{zbn}
    \begin{array}{lllll}
          z^{b(m)}_{tt}(x,t)=
         \Delta  z^{b(m)}(x,t)+f( z^{b(m)}(x,t), x,t),  \\
         z^{b(m)}(x,t)  =
          0,\quad x \in \Gamma_D, \\
         \frac{\partial}{\partial \nu} z^{b(m)}(x,t)  = -k[ y(x,t)- z_t^{b(m)}(x,t) ], \quad x \in \Gamma_N,
            \\  \ t\in[t_0,t_0+T],\\
        z^{b(m)}(x,t_0+T)=z^{(m)}(x,t_0+T),\
        \\
         z_t^{b(m)}(x,t_0+T)= z_t^{(m)}(x,t_0+T). \\
    \end{array}
\end{equation}
This results in the sequence of the forward $e^{(m)}=z-z^{(m)}$ and
the backward  $e^{b(m)}=z-z^{b(m)}, \ m=1,2,...$ errors satisfying
\begin{equation}\label{en}
    \begin{array}{l}
          e^{(m)}_{tt}(x,t)=
         \Delta  e^{(m)}_{x}(x,t)+g^{(m)}e^{(m)}(x,t),  \\
        e^{(m)}(x,t)\! \Big\rvert_{x \in \Gamma_D}\!  = \! 0,
        \quad \frac{\partial}{\partial \nu} e^{(m)}(x,t) \! = \! - \! k  e_t^{(m)}(x,t)\! \Big\rvert_{x \in \Gamma_N},
        \\ \quad t\in[t_0,t_0+T],\\
        e^{(m)}(x,t_0)= e^{b(m-1)}(x,t_0),\\
        e_t^{(m)}(x,t_0)=  e_t^{b(m-1)}(x,t_0),\\
    \end{array}
\end{equation}
where $ e^{b(0)}(x,t_0)=z_0(x),\   e_t^{b(0)}(x,t_0)=z_1(x)$
   and
\begin{equation}\label{ebn}
    \begin{array}{l}
    e^{b(m)}_{tt}(x,t)=\Delta e^{b(m)}(x,t)+ g^{b(m)}e^{b(m)}(x,t),  \\
    e^{b(m)}(x,t)\! \Big\rvert_{x \in \Gamma_D} \! = \! 0, \quad \frac{\partial}{\partial \nu} e^{(m)}(x,t) \! =  \! k e_t^{(m)}(x,t)     \! \Big\rvert_{x \in \Gamma_N}, \\
    \quad t \in[t_0,t_0+T],\\
    e^{b(m)}(x,t_0+T)=e^{(m)}(x,t_0+T),
    \\
    e_t^{b(m)}(x,t_0+T)=e_t^{(m)}(x,t_0+T).
    \end{array}
\end{equation}
Here
\begin{equation}
    \begin{array}{llll}\label{gm}
        g^{(m)} =g( z, e^{(m)},x, t)=\int_0^1 f_{z}(z+ (\theta-1)  e^{(m)}, x,t) d\theta, \\
        g^{b(m)}\!=\!g( z, e^{b(m)},x, t)\!=\! \int_0^1 f_{z}( z + (\theta-1)e^{b(m)}, x,t) d\theta.
    \end{array}
\end{equation}

\subsection{LMIs  for the exact observability time} 

For (\ref{en}) and  (\ref{ebn}) we  consider for $t\in[t_0, t_0+T]$
the Lyapunov functions
\begin{equation}
    \begin{array}{llll}\label{Vn}
   V^{(m)}(t)   = E^{(m)}(t)
            +\!\chi \frac{k(n-1)}{2}\int_{\Gamma_N}(e^{(m)})^2 d\Gamma \\
                         \quad +  \chi  \int_{\Omega}
            \left[ 2  (x^T  \cdot  \nabla e^{(m)} )  +  (n  -  1)e^{(m)} \right] {e^{(m)}_t}  dx ,\\
   \quad \!E^{(m)}(t)\! = \! \frac{1}{2} \int_{\Omega}
            \left[|\nabla e^{(m)}|^2+  ({e^{(m)}_t})^2 \right]
            dx
    \end{array}
\end{equation}
  and

\begin{equation}
    \begin{array}{llll}\label{Vbn}
   V^{b(m)}(t)  =   E^{b(m)}(t)
                                +\chi \frac{k(n-1)}{ 2}\int_{\Gamma_N}(e^{b(m)})^2 d\Gamma \\
                        \quad - \!  \chi  \int_{\Omega}
                         \left[ 2  (x^T  \cdot  \nabla e^{b(m)} ) +  (n-1)e^{b(m)} \right] {e^{b(m)}_t}  dx,\\
 E^{b(m)}(t) \! \!= \! \frac{1}{2} \int_{\Omega}
                        \left[|\nabla e^{b(m)}|^2  +  ({e^{b(m)}_t})^2\right]
                         dx
    \end{array}
\end{equation}
  with some constant  $\chi>0$. Then for $\chi$ and $\lambda_0>0$ subject to \eqref{Phi0}
  we have (cf. (\ref{Vpos}))
\begin{equation}
    \begin{array}{llll}\label{Vn>alpha}
    \alpha E^{(m)}(t)  & \leq  V^{(m)}(t)  & \leq  \beta E^{(m)}(t), \quad t\geq t_0,\\
    \alpha E^{b(m)}(t) & \leq  V^{b(m)}(t) & \leq  \beta E^{b(m)}(t),
    \end{array}
\end{equation}
where $\alpha$ and $\beta$ are given by \eqref{Vpos}.
\begin{lemma}\label{lem_en}
Consider $V^{(m)}$ and $V^{b(m)}$ given by (\ref{Vn}) and
(\ref{Vbn}) respectively with  $\chi>0$ satisfying
\eqref{Phi0}. Assume there exist $\delta>0$ and $ T>0$  such that for
all $m=1,2,...$ and for all $ t\in[t_0, t_0+T]$ the 
inequalities
\begin{equation}\label{lem_ena}
    \dot V^{(m)}(t)+2\delta V^{(m)}(t)\leq 0
\end{equation}
and
\begin{equation}\label{lem_enab}
    \dot V^{b(m)}(t)-2\delta V^{b(m)}(t)\geq 0
\end{equation}
hold along (\ref{en}) and (\ref{ebn})
respectively. Assume additionally that for some  $T^*\in(0,T)$
\begin{equation}
    \begin{array}{llll}\label{lem_enb}
        V^{(m)}(t_0)e^{-2\delta T^*} & \leq &  V^{b(m-1)}(t_0),\\
             V^{b(m)}(t_0+T)e^{-2\delta T^*} & \leq & V^{(m)}(t_0+T).
    \end{array}
\end{equation}
 Then  the iterative algorithm
    converges on $[t_0, t_0+T]$:
\begin{equation}\label{lem_enThena}
    V^{b(m)}(t_0)\leq q V^{b(m-1)}(t_0)\leq q^{m} V^{b(0)}(t_0) ,
\end{equation}
$q=e^{-4\delta (T-T^*)}$ is the convergence rate.

Moreover, for all $t\in [t_0, t_0+T]$ and $m=1,2,...$
\begin{equation}\label{lem_reg}\begin{array}{ll}
\max\{V^{(m)}(t), V^{b(m)}(t)\}\leq e^{2\delta T^*} V^{b(0)}(t_0).
\end{array}\end{equation}
\end{lemma}
\vspace{-0.2cm}
\begin{proof}
The inequalities (\ref{lem_ena}), (\ref{lem_enab}) yield
$$\begin{array}{ll}
V^{b(m)}(t_0)\!\leq \!V^{b(m)}(t_0\!+\!T)\!e^{-2\delta T}\!, \\
V^{(m)}(t_0\!+\!T)\!\leq\!
V^{(m)}(t_0)\!e^{-2\delta T}.\end{array}$$ 
 Hence, (\ref{lem_enb})
implies \eqref{lem_enThena}:
$$\begin{array}{ll}
 V^{b(m)}(t_0)\leq
V^{b(m)}(t_0+T)e^{-2\delta T}\leq\!\!V^{(m)}(t_0\!+\!T)\sqrt{q}\\ 
\leq V^{(m)}(t_0)\sqrt{q}e^{-2\delta  T}
\leq V^{b(m-1)}(t_0)q. \end{array}$$
The bound (\ref{lem_reg}) follows from the following inequalities:
$$\begin{array}{llllll} V^{(m+1)}(t)\leq   V^{(m+1)}(t_0)\leq e^{2\delta T^*} V^{b(m)}(t_0)\\
\leq
V^{b(m)}(t_0+T)\leq V^{(m)}(t_0+T)e^{2\delta T^*}\\ \leq
V^{(m)}(t_0)\leq ... \leq V^{(1)}(t_0)
\leq e^{2\delta T^*}V^{b(0)}(t_0),\\
V^{b(m)}(t)\leq V^{b(m)}(t_0+T) 
\leq e^{2\delta T^*}V^{b(0)}(t_0).\end{array}$$
\vspace{-0.2cm}
\end{proof}
We are in a position to formulate sufficient conditions
for the exact observability:

\begin{theorem}\label{Prop_ExactObs}
 Given positive tuning parameters $T^*$ and $
\delta $, let there exist positive constants $\chi$, $\lambda_1$ and
 $\lambda_2$ that satisfy the LMIs (\ref{Psi1_4}) and
\begin{equation}
    \begin{array}{l}\label{Phi}
         \Phi\bydef
           \begin{bmatrix}
           \Phi_{11} & \sqrt{n}[1+e^{-2\delta T^*}]\chi & 0 \\
                {*}       & -\frac{1}{ 2}[1-e^{-2\delta T^*}] & \frac{n-1}{2}[1+e^{-2\delta T^*}]\chi         \\
                {*}       & {*}        & -\lambda_2
            \end{bmatrix}<0. \\
         \Phi_{11}= -\frac{1}{2}[1-e^{-2\delta T^*}]+ \lambda_2 \frac{4}{\pi^2 n}
    \end{array}
\end{equation}
Then
\newline
(i) the
system (\ref{Wave_nd_z_bound})-(\ref{Wave_nd_z_initial})
with the measurements (\ref{Wave_nd_z_meas})  is exactly observable in time $T^*$;
\newline
(ii) for all $\Delta T>0$
 the iterative algorithm with $T=T^*+\Delta T$
converges
\begin{equation}\label{thm_end}
    \begin{array}{ll}
        \int_{\Omega}\Big[
        |\nabla e^{b(m)}(x,t_0)|^2+[e^{b(m)}_t(x,t_0)]^2\Big]dx\\ \leq
        \frac{\beta}{ \alpha} q^m \int_{\Omega}\Big[ |\nabla z_{0}|^2(x)+z_1^2(x)\Big]dx,
    \end{array}
\end{equation}
where  $q=e^{-4 \delta  \Delta T}$,
and the following bound holds:
\begin{equation}\label{thm_ebound}
    \begin{array}{ll}
       max\Big \{\int_{\Omega}\left[
        |\nabla e^{b(m)}(x,t)|^2+[e^{b(m)}_t(x,t)]^2\right]dx,\\
        \int_{\Omega}\left[
        |\nabla e^{(m)}(x,t)|^2+[e^{(m)}_t(x,t)]^2\right]dx\Big \}
        \\ \leq
        \frac{\beta}{ \alpha} e^{2\delta T^*} \int_{\Omega}\Big[ |\nabla z_{0}|^2(x)+z_1^2(x)\Big]dx,\\
         t\in[t_0, t_0+T].
    \end{array}
\end{equation}
 Here  $\alpha$ and $\beta$
are given by (\ref{Vpos}).
\end{theorem}
\begin{proof}
(i) From Theorem \ref{wave_nd_exp} it follows that LMIs
\eqref{Psi1_4} yield  (\ref{lem_ena}).  By the similar derivations,
LMIs \eqref{Psi1_4} imply  (\ref{lem_enab}) for the backward system.
Taking into account that $e^{(m)}(x,t_0+T)=e^{b(m)}(x,t_0+T)$ and
$e^{(m)}_t(x,t_0+T)=e^{b(m)}_t(x,t_0+T)$, the bound \eqref{Vbound2} and the n-D Wirtinger inequality
we obtain for some $\lambda_2>0$
\begin{equation*}
    \begin{array}{l}
         V^{b(m)}(t_0+T)e^{-2\delta T^*}-V^{(m)}(t_0+T)\\
        \quad =
        \frac{1}{ 2}[-1+e^{-2\delta T^*}]\Big\{\int_{\Omega} [(e^{(m)}_t)^2\\
        \quad +\vert\nabla
        e^{(m)}\vert^2]dx+\chi k(n-1)\int_{\Gamma_N} (e^{(m)})^2 d \Gamma\Big\}\\
        \quad  \! - \!  \chi [1+e^{ \! - \! 2\delta T^*}]  \int_{\Omega}
        \left[ 2  (x^T \cdot \nabla e^{(m)} ) +  (n-1)e \right] e_t^{(m)}  dx\\
        \quad \leq \frac{1}{ 2}[-1+e^{-2\delta T^*}]\int_{\Omega} [(e_t^{(m)})^2+\vert\nabla
        e^{(m)}\vert^2]dx\\
        \quad +\chi [1+e^{-2\delta T^*}] \int_{\Omega} \Big[ 2\sqrt{n}|\nabla e^{(m)}| + \\
        \quad +(n-1)|e^{(m)}|\Big] |e^{(m)}_t |\!dx\\
        \quad  \! + \! \lambda_2\int_{\Omega} \left[ \frac{4}{\pi^2 n} |\nabla e^{(m)}|^2 \! - \! (e^{(m)})^2 \right] dx
         \! \leq  \! \int_{\Omega}\eta_2^T\Phi\eta_2dx \! \leq  \!  0,
    \end{array}
\end{equation*}
where
\begin{equation}\label{eta2m}
    \begin{array}{l}
        \eta_2=
        col \{|\nabla e^{(m)}(x,t)|,|e_{t}^{(m)}(x,t)|,|e^{(m)}(x,t)|\} 
    \end{array}
\end{equation}
and where $t=t_0+T$,
 if \eqref{Phi} is feasible.
Similarly \eqref{Phi} guarantees $V^{(m)}(t_0)e^{-2\delta T^*}\leq V^{b(m-1)}(t_0)$.
The feasibility of the LMI \eqref{Phi} yields the feasibility of
 \eqref{Phi0}, i.e. the positivity of
$V^{(m)}$ and $V^{b(m)}$.
Moreover, the strict LMI \eqref{Phi} guarantees \eqref{lem_enb} with $T^*$
changed by $T^*-\Delta T$, where $\Delta T>0$ is small enough,
implying due to Lemma \ref{lem_en} the convergence of the iterative algorithm with $T=T^*$.

To prove the exact observability in time $T^*$, consider  initial states
$\zeta(t_0)\in {\mathcal H}$ and  $\bar
\zeta(t_0)\in {\mathcal H}$ of (\ref{Wave_nd_z_bound})-(\ref{Wave_nd_z_initial}) that lead to the measurements $y(x,t)$ and $
\bar y(x,t)$ and to the corresponding forward and backward  observers $z^{(m)}, z^{b(m)} $ and $ \bar
z^{(m)}, \bar z^{b(m)}$. Note that  $ \bar
z^{(m)}, \bar z^{b(m)}$ satisfy (\ref{zn}) and (\ref{zbn}), where $z^{(m)}, z^{b(m)} $  and $y$ are replaced by $ \bar
z^{(m)}, \bar z^{b(m)}$ and $\bar y$.
The resulting $ e^{(m)}=z^{(m)}-\bar z^{(m)}$, $ e^{b(m)}=z^{b(m)}-\bar z^{b(m)}$
satisfy (\ref{en}), (\ref{ebn}) with the perturbed boundary conditions at $x\in \Gamma_N$:
\begin{equation}\label{bcw}
    \begin{array}{l}
        \frac{\partial}{\partial\nu}e^{(m)} \! =  \! -ke_t^{(m)}+w, \; w\bydef k[y(x,t)-\bar y(x,t)],\\
        \frac{\partial}{\partial\nu}e^{b(m)} \! =  \! ke_t^{b(m)}-w, \; x\in \Gamma_N,\ t\geq t_0.
    \end{array}
\end{equation}
Let $V^{(m)}$ and $V^{b(m)}$ be defined by (\ref{Vn}) and
(\ref{Vbn}). LMI (\ref{Phi}) implies inequalities
(\ref{lem_enb}).

We will show next that the feasibility of \eqref{Psi1_4}
 implies
\begin{equation}
    \begin{array}{l}
        \label{ISS}
        \dot { V}^{(m)}(t) +2\delta  V^{(m)}(t)-\gamma \int_{\Gamma_N}w^2d\Gamma \leq 0,\\
        \dot { V}^{b(m)}(t) -2\delta  V^{b(m)}(t)+\gamma \int_{\Gamma_N}w^2d\Gamma\geq 0
    \end{array}
\end{equation}
for $t\geq t_0$ and some $\gamma>0$.
Taking into account $w$-term  in \eqref{bcw},
by the arguments of Theorem \ref{wave_nd_exp}
 we have
\begin{equation*}
    \begin{array}{l}
       \dot{E}^{(m)}(t) =  \int_{\Gamma_N}\left[-k \left(e_t^{(m)}\right)^2+e_t^{(m)}w\right] d \Gamma \\
        \quad + g_1 \int_\Omega |e^{(m)}| |e_t^{(m)}| dx,
     \end{array}
\end{equation*}
    and
\begin{equation*}
    \begin{array}{l}
     \frac{d}{dt} \left \{ \int_{\Omega} [2x^T\nabla e^{(m)}
    +(n-1) e^{(m)}] e_t^{(m)}  dx\right\}\\
     \quad =-\int_{\Omega}\{ (e_t^{(m)})^2+\vert\nabla
    e^{(m)}\vert^2 + [2x^T\nabla e^{(m)}\\
    \quad +(n-1) e^{(m)}] g e^{(m)} \} dx
    \\
    \quad-\int_{\Gamma_N}\left[\vert\nabla e^{(m)})\vert^2 +2 x^T\nabla e^{(m)}[ke_t^{(m)}-w]\right ] d\Gamma\\
    \quad +\int_{\Gamma_N}\left[(e_t^{(m)})^2-(n-1)e [k e_t^{(m)}-w]\right]d\Gamma.
     \end{array}
\end{equation*}
Then after bounding and completion of squares we find
\begin{equation*}
    \begin{array}{l}
          \frac{d}{dt} V^{(m)}(t)\! +\! 2 \delta V^{(m)}(t) \\
        \quad \leq
          \Psi_1 \int_{\Gamma_N}  (e_t^{(m)})^2 d \Gamma
          \! +  \! \int_\Omega  \!\eta_2^T \!\Psi_2  \eta_2 dx \\
         \quad  +\int_{\Gamma_N}\Big\{|e_t^{(m)}||w|+\chi(n-1)|e^{(m)}||w|+ \\
         \quad + \chi k^2 n\Big[2 |e^{(m)}_t||w|+w^2\Big]\Big\} d \Gamma
         \leq 0,
     \end{array}
\end{equation*}
where $\eta_2$ is given by \eqref{eta2m}.
By Young's inequality with some $r>0$ and by \eqref{lambda2}
\begin{equation*}
    \begin{array}{l}
        \chi(n-1)\int_{\Gamma_N}|e^{(m)}||w| d \Gamma\leq
        \frac{\chi(n-1) }{ 2r}\int_{\Gamma_N}(e^{(m)})^2d\Gamma\\
        +
        {\chi(n-1)r\over 2}\int_{\Gamma_N}w^2d\Gamma\\
        \leq \frac{\chi(n-1) }{ 2r}\int_{\Omega}|\nabla e^{(m)}|^2dx+
        {\chi(n-1)r \over 2}\int_{\Gamma_N}w^2d\Gamma
    \end{array}
\end{equation*}
Then the first inequality \eqref{ISS} holds if
\begin{equation}\label{ISS0}
    \begin{array}{l}
       \begin{bmatrix}
            \Psi_1 & \chi\Big(\frac{1}{ 2}+k^2n\Big)\\
            {*} &-\gamma+\chi k^2 n+{\chi(n-1)r \over 2}
        \end{bmatrix} <0, \\
       \Psi_2+\frac{\chi(n-1) }{ 2r}[1 \ 0 \ 0]^T[1 \ 0 \ 0]<0.
    \end{array}
\end{equation}
It is easy to see that
the latter inequalities are feasible for large
enough $r$ and $\gamma$ if $\Psi_1<0$ and $\Psi_2<0$, i.e. if  LMIs \eqref{Psi1_4} are
satisfied. Then, by the comparison principle (see e.g. \cite{Khalil}),
\begin{equation*}\label{ISSV}
    \begin{array}{ll}
    V^{(m)}(t)\leq e^{ \! - \! 2\delta (t \! - \! t_0)}V^{(m)}(t_0) \! + \! \gamma \int^{t}_{t_0}\int_{\Gamma_N}|w(x,s)|^2d\Gamma ds.
    \end{array}
\end{equation*}

Similarly, LMIs (\ref{Psi1_4}) guarantee the second inequality \eqref{ISS}
 for  large enough
$\gamma>0$, and, thus,
\begin{equation*}
    \begin{array}{l}
        V^{b(m)}(t)\geq e^{2\delta (t-t_0)}V^{b(m)}(t_0)\\
        -\gamma
        \int^{t}_{t_0}\!\int_{\Gamma_N}\!e^{2\delta (t-s)}|w(x,s)|^2d\Gamma ds.
    \end{array}
\end{equation*}
Note that the strict inequalities \eqref{Phi} guarantee
\eqref{lem_enb}
with $\delta$ changed by $\delta+\delta_0$ for small enough $\delta_0>0$.
Therefore,
\begin{equation*}
    \begin{array}{llllll} V^{b(m)}(t_0)\leq
    e^{-2(\delta+\delta_0) T^*} V^{b(m)}(t_0+T^*)\\ \quad+\gamma \int^{t_0+T^*}_{t_0}\int_{\Gamma_N}|w(x,s)|^2d\Gamma ds\\
    \leq
    e^{-2\delta_0 T^*}V^{(m)}(t_0+T^*)\\
    \quad+\gamma \int^{t_0+T^*}_{t_0}\int_{\Gamma_N}|w(x,s)|^2d\Gamma ds\\
    \leq
    e^{-4\delta_0 T^*} V^{b(m-1)}(t_0)\\ \quad +(e^{-2\delta_0 T^*}+1)\gamma \int^{t_0+T^*}_{t_0}|\int_{\Gamma_N}|w(x,s)|^2d\Gamma ds.
    \end{array}
\end{equation*}
We arrive at
\begin{equation*}
    \begin{array}{llll}
        \alpha\int_0^{1}\Big[
        [e^{b(m)}_{x}(x,t_0)]^2+[e^{b(m)}_t(x,t_0)]^2\Big]dx\\
        \quad \leq
         V^{b(m)}(t_0)\leq (e^{-4\delta_0 T^*})^2 V^{b(m-2)}(t_0)\\
         \quad +(e^{\!-\!6\delta_0 T^*}\!+\!e^{\!-\!4\delta_0 T^*}\!+\!e^{\!-\!2\delta_0 T^*}\!+\!1)\gamma \!\int^{t}_{t_0}\!\int_{\Gamma_N}|w(s)|^2d\Gamma ds\\
         \quad \leq
         e^{-4m\delta_0 T^*} V^{b(0)}(t_0)+
         \alpha C \int^{t_0+T^*}_{t_0}\int_{\Gamma_N}|w(s)|^2d\Gamma ds
     \end{array}
\end{equation*}
 which implies (\ref{def_exact}), where
$C=\frac{\gamma}{ \alpha[1-e^{-2\delta_0 T^*}]}$.

(ii) follows from \eqref{lem_enThena}, \eqref{lem_reg} and \eqref{Vn>alpha}.
\end{proof}

\begin{remark}\label{rem_ISS}
As a by-product, we have derived new LMI conditions \eqref{ISS0}
for  input-to-state stability of the n-D wave equation \eqref{en} with
the perturbed boundary condition on $\Gamma_N$ as in \eqref{bcw}.
\end{remark}

\begin{remark}
Note that for $n=1$ and $g=0$ the LMIs of Theorem \ref{Prop_ExactObs}
are equivalent to the corresponding conditions of \cite{fridman2013observers} that
are not conservative
(in the sense that they lead to the analytical value of  the minimal observability time $T^*_{an}$).
However, for $n=2$ and $g=0$ the conditions of
Theorem \ref{Prop_ExactObs}
lead to an upper bound on $T^*_{an}$ only 
(see Example \ref{ex_2D} below).
This mirrors the conservatism of the conditions for $n>1$.
\end{remark}

\begin{example}\label{ex_2D}
Consider \eqref{Wave_nd_z_bound}-\eqref{Wave_nd_z_initial}, where $n=2$ with the values of $g_1$ as given in Table \ref{table:nonlin}.
 We use
  the sequence of forward and backward observers \eqref{zn} and \eqref{zbn}
with $k=1$.
By verifying the conditions  of Theorem \ref{Prop_ExactObs},  we find
 the minimal values of $T^*$ and the corresponding $\delta$ for the convergence
of the iterative algorithm and, thus, for the exact observability.
Note that for $g_1 = 0$ the observability time is $T^* = 3.28$ , which is not too far
from the analytical value
 $2 \sqrt{2}\approx 2.82 $.
 For simulation results in the linear case see Example 2 of \cite{GeorgeAut10}.

\begin{table}[ht]
\caption{Nonlinearity vs. minimal observability time}
\centering
\begin{tabular}{ c  c  c  }
\hline \hline
$g_1$ & $\delta$ & $T^*$  \\ [0.5ex]
\hline
0 & 0.0001 & 3.28 \\
0.01  & 0.01 & 4.3 \\
0.1   & 0.01 & 12.2  \\
0.3 & 0.01 & 38 \\ [1ex]
\hline
\end{tabular}
\label{table:nonlin}
\end{table}

\end{example}

\section{Regional observability of  1-D wave equation with locally Lipschitz nonlinearity }\label{regional_obs_section}

In this section we consider 1-D wave equation \eqref{Wave_nd_z_bound}, where $\Omega=[0,1]$:
\begin{eqnarray}\label{Wave_nd_z_nl_bound}
    \nonumber  z_{tt}(x,t) =z_{xx}(x,t) + f(z,x,t),  &   \qquad x\in [0,1], \ t>t_0, \\
    z(0,t)  = 0, \qquad z_x(1,t)=0, &
\end{eqnarray}
whereas the measurements are given by
\begin{equation}\label{1dmeas}
y(t)=z_t(1,t), \quad t\in[t_0, t_0+T].
\end{equation}

 Assume that $f(0,x,t)\equiv 0$ and that $f$ is {\it locally Lipzchitz} in the first argument
  uniformly on the  others.
The latter means that we can find a $ d>0 $ such that
\begin{equation}\label{locL}
  | f_z | \leq g_1 \qquad \forall |z|\leq d,\ x\in[0,1],\ t\geq t_0 .
\end{equation}
We present
 \begin{equation}\begin{array}{ll}
\label{f1}
 f(z,x,t)=f_1 z, \quad f_1=\int_0^1 f_{z}(\theta z,x,t) d\theta.\\
 \end{array}\end{equation}

 Recall that in 1-D case
${\mathcal
H}={\mathcal H}_{\Gamma_D}^{1}(0,1)\times L^2(0,1)$, where
$${\mathcal H}^{1}_{\Gamma_D}(0,1) =\left\{\zeta_0\in {\mathcal H}^{1}(0,1)\ \bigg| \ \zeta_0(0)=0\right \}$$
and
$$ 
\begin{array}{lll} {\mathcal
D}({\mathcal A})=
\Big\{&\!\!(\zeta_0,\zeta_1)^T\in \!{\mathcal H}^{2}(0,1)\!\bigcap{\mathcal
H}_{\Gamma_D}^{1}(0,1)\!\times \!{\mathcal H}_{\Gamma_D}^{1}(0,1) \\
& \bigg|\
\zeta_{0x}(1)=0 \Big\}.
\end{array}$$
Consider a region of initial conditions defined by
\begin{equation}\label{region}
    \begin{array}{l}
    {\mathcal X}_{d_0}= \left\{\Big[z_0,z_1\Big]^T\in {\mathcal H}\ \bigg|  \int_0^1\left [{z_0}^2_x  + z_1^2 \right]dx \leq
        d_0^2 \right\},
   \end{array}
\end{equation}
where $d_0>0$ is some constant.
 We are looking for an estimate ${\mathcal X}_{d_0}$ (with $d_0$ as large as possible) on the region of initial conditions,
    for which the iterative
  algorithm defined in
    Section \ref{sec_ExObs}
   converges. This gives an estimate on the region of exact observability, where the initial conditions of the system can be recovered
 uniquely from the measurements on the interval $[t_0,t_0+T]$.

 The convergence of the iterative algorithm in Theorem \ref{Prop_ExactObs}
   has been proved for the  forward  and the backward
   error systems \eqref{en} and \eqref{ebn} with globally Lipschitz
   nonlinearities given by \eqref{gm}
   subject to
\begin{equation}
    \begin{array}{l}
         |f_{z}(z + (\theta-1)  e^{(m)} ,x,t)|  \leq g_1, \\
         |f_{z}(z + (\theta-1)  e^{b(m)},x,t)|  \leq g_1, \\
         \quad \forall t\in[t_0, t_0+T],\; x,  \theta\in [0,1], \;  z, e^{(m)}, e^{b(m)}\in \bl R.
    \end{array}
\end{equation}
For the locally bounded nonlinearity as in \eqref{locL} we have  to find a region $ {\mathcal X}_{d_0}$
of initial conditions  starting from which  solutions of \eqref{Wave_nd_z_nl_bound}, \eqref{en} and
\eqref{ebn} satisfy the bound
\begin{equation}
    \begin{array}{l}\label{gml}
[z_0, z_1]^T\in  {\mathcal X}_{d_0}\Rightarrow    |z(x,t)+(\theta - 1) e^{(m)}(x,t) | \leq  d,\\
   \hspace{2.8cm} |z(x,t)+(\theta - 1) e^{b(m)}(x,t) | \leq  d,\\
    \hspace{2.8cm}  \forall t\in[t_0, t_0+T],\ x, \theta\in [0,1].
    \end{array}
\end{equation}
The latter implication yields
\begin{equation}\label{g1f1}
    \begin{array}{l}
[z_0, z_1]^T\in  {\mathcal X}_{d_0}\Rightarrow \max\{|f_1|, |g^{(m)}|,
  |g^{b(m)}|\}\leq g_1, \\
  \forall t\in[t_0, t_0+T],\ x, \theta\in [0,1]
    \end{array}
\end{equation}
We will employ Sobolev's inequality
\begin{equation}\label{1Dbound}
    \begin{array}{ll}
max_{x\in[0,1]}z^2(x,t)\leq \int_0^1z_x^2(x,t)dx,\ t\geq t_0 
    \end{array}
\end{equation}
that holds since $z\Big\rvert_{x=0}=0$, and similar bounds on $e^{(m)}$ and $e^{b(m)}$.
In order to guarantee \eqref{gml} we  start with a bound on the solutions
of \eqref{Wave_nd_z_nl_bound}.
 Since this system is not stable we  give a simple energy-based
  bound
  on the exponential growth of $z$.
Define the energy 
\begin{equation*}\label{Lyapunov_nd_z}
    \begin{array}{ll}
    E_{z_{eq}}(t) =
    \frac{1}{2}\int_{0}^1
    \Big( z_x^2  + z_t^2 \Big)
    dx.
    \end{array}
\end{equation*}

\begin{proposition}\label{exp_growth_bound}
    Consider \eqref{Wave_nd_z_nl_bound} with $f(0,x,t)\equiv 0$
    subject to
    $|f_z|\leq g_1$ for all $(z,x,t)\in \bl R^{3}$.
    Then solutions of this system satisfy  the following inequality:
\begin{equation*}
\label{Vzbound}E_{z_{eq}}(t)\leq e^{\frac{2g_1}{\pi}
(t-t_0)}E_{z_{eq}}(t_0), \quad t\geq t_0.\end{equation*}
\end{proposition}
\begin{proof}
It is sufficient to show that
\begin{equation*}\begin{array}{l}\label{dEz}
W\bydef\dot E_{z_{eq}}-\frac{2g_1}{ \pi} E_{z_{eq}}\leq 0
\end{array}\end{equation*}
 along (\ref{Wave_nd_z_nl_bound}). Differentiating,
 integrating by parts, taking into account the boundary conditions (that imply $z_x(1,t)=z_t(0,t)=0$) and further applying Wirtinger's inequality we have
 \begin{equation*}
    \begin{array}{l}
     W=\int_0^1[z_x z_{xt}+z_t(z_{xx}+f)]dx-\frac{2g_1}{ \pi} E_{z_{eq}}\\
     \quad=
     \int_0^1 z_tf_1zdx-\frac{2g_1}{ \pi} E_{z_{eq}}\\
      \quad  \leq g_1\int_0^1|z_t||z|dx-\frac{g_1}{ \pi}\int_{0}^1
        \Big( \frac{\pi^2}{ 4}z^2  + z_t^2 \Big)dx\\
       \quad =-\frac{g_1}{ \pi}\int_{0}^1
        \Big( \frac{\pi}{ 2}|z|  - |z_t| \Big)^2dx\leq 0.
    \end{array}
\end{equation*}
\end{proof}
Due to \eqref{1Dbound},
given $d>0$ the solution $z$ of  \eqref{Wave_nd_z_nl_bound} satisfies the bound
\begin{equation}\label{zd0} z^2(x,t) \leq  0.25 {d^2} \quad \forall x \in [0,1] , \ t \in [t_0, t_0 +T] \end{equation}
 if
\begin{equation}\begin{array}{ll}
\label{leq_d1}
\max_{x\in[0,1]} z^2(x,t) \leq \int_{0}^1 \Big[ z_x (x,t)^2  + z_t(x,t)^2 \Big] dx
 \\
 \quad \leq  e^{\frac{2g_1}{ \pi}
(t-t_0)} \int_{\Omega} \Big[ | z_0(x) |^2  + z_1(x)^2 \Big] dx \leq  \frac{d^2}{4}.
\end{array}\end{equation}
In order to bound $e^{(m)}$ and $e^{b(m)}$, we use  Theorem \ref{Prop_ExactObs}.
The LMIs \eqref{Psi1_4}  for $n=1$ are reduced to
\begin{equation}\begin{array}{l} \label{psi1_Reg}
         -k+(1+k^2)\chi<0,\\
            \begin{bmatrix}
           -\chi+\delta+ \lambda_1 \frac{4}{\pi^2 }& 2\delta \chi &  g_1\chi\\
            * & -\chi+\delta & \frac{1}{ 2}g_1\\
            * & * & -\lambda_ 1
            \end{bmatrix}
            \leq 0,
            \end{array}
\end{equation}
where $\chi$ and $\lambda_1$ are positive scalars.
The LMI \eqref{Phi} for $n=1$ has a form
\begin{equation}\label{phi_Reg}
       \begin{bmatrix}
       -\frac{1}{2}[1-e^{-2\delta T^*}]& [1+e^{-2\delta T^*}]\chi \\ 
            {*}       & -\frac{1}{ 2}[1-e^{-2\delta T^*}]
        \end{bmatrix}<0.
\end{equation}
The LMI \eqref{Phi0}
has a form $\Phi_0>0$, where $2\Phi_0=\begin{bmatrix}{1} & 2\chi\\
{*} & {1}\end{bmatrix}$, leading to
$\alpha=2\lambda_{min}(\Phi_0)$
and $\beta=2\lambda_{max}(\Phi_0)$ in the bounds \eqref{Vn>alpha}.
Hence,
    $\alpha=(1-2\chi)$ and $\beta=(1+2\chi)$.

Similarly to \eqref{leq_d1},
 if the LMIs
\eqref{psi1_Reg} are feasible, then 
\begin{equation*}
    \begin{array}{ll}
    \max \Big\{[e^{(m)}(x,t)]^2, [e^{b(m)}(x,t)]^2 \Big \}\leq \frac{d^2}{ 4}
    \end{array}
\end{equation*}
$ \forall x \in [0,1] , \ t \in [t_0, t_0 +T]$ provided (cf. \eqref{thm_ebound})
\vspace{-0.3cm}
\begin{equation}
    \begin{array}{ll}\label{leq_d2}
     max\Big\{\max_{x\in[0,1]}[e^{(m)}(x,t)]^2, \max_{x\in[0,1]}[e^{b(m)}(x,t)]^2 \Big\}
     \\
      \leq \frac{1+2\chi}{1-2\chi} e^{2 \delta T^*}  \int_{0}^1 \Big[  z_{0x}^2(x)   + z_1^2(x) \Big] dx \leq \frac{d^2}{ 4}.
    \end{array}
\end{equation}

    Denote
\begin{equation}\label{d_0}
d_0 \bydef \frac{d}{ 2}\cdot min\Big\{ e^{-\frac{g_1}{ \pi}T}
 , \quad   \sqrt{\frac{1-2\chi}{1+2\chi }}e^{-\delta T^*}\Big\}.
\end{equation}
Then due to   \eqref{leq_d1}
for all solutions of \eqref{Wave_nd_z_nl_bound}
initiated from \eqref{region} the bound \eqref{zd0} holds.
Moreover,
   due to \eqref{leq_d2} for all the resulting $e^{(m)}(x,t)$ and  $e^{b(m)}(x,t)$
   that satisfy \eqref{en} and \eqref{ebn} respectively the implication \eqref{gml} holds:
\begin{equation*}
\begin{array}{lll}
    |z(x,t)+(\theta - 1) e^{(m)}(x,t) |^2 \\
    \quad \leq 2 z^2(x,t) + 2 [e^{(m)}(x,t)]^2 \leq d^2,\\
    |z(x,t)+(\theta - 1) e^{b(m)}(x,t) |^2 \\
    \quad \leq 2 z^2(x,t) + 2 [e^{b(m)}(x,t)]^2 \leq d^2,\\
    \forall t\in[t_0, t_0+T],\ x\in [0,1], \theta\in [0,1].
    \end{array}
\end{equation*}
The latter bounds guarantee \eqref{g1f1}.
Then from Theorem \ref{Prop_ExactObs}
we conclude
 the following:
\begin{corollary} \label{prop_reg_0bs_1d}
    Given $g_1$ and positive tuning parameters $T^*$ and $\delta$ ,
    let there exist positive constants $\chi $ and $\lambda_1$  that satisfy the LMIs
  \eqref{psi1_Reg} and \eqref{phi_Reg}.
    Then for all $T\geq T^*$
    the system \eqref{Wave_nd_z_nl_bound} subject to $f(0,x,t)\equiv 0$ and \eqref{locL} with
     the measurements \eqref{1dmeas} is {\it regionally exactly observable} on $[t_0,t_0+T]$ for all initial conditions from ${\mathcal X}_{d_0}$ given by \eqref{region},
where   $ d_0$ is
  defined by  \eqref{d_0}.
\end{corollary}

\begin{remark}\label{rem_reg}
The result on the regional observability
cannot be extended to multi-dimensional case
since the bound \eqref{1Dbound} does not hold in n-D case.
One could extend the regional result to  n-D case  if $f$ would depend on
$\int_{\Omega}|\nabla z|^2dx$ or on
$\int_{\Omega}z^2dx$, $\int_{\Gamma_N}z^2d\Gamma$ (by employing the inequalities of Lemma \ref{lem_ineq}).

The global results of Sections \ref{Observers} and  \ref{sec_ExObs} can be extended to more general functions $f=f(z,\nabla z, z_t)$ with uniformly bounded $f_z, |f_{\nabla z}|$ and $f_{z_t}$. Note that in \cite{fridman2013observers} such more general functions were considered for 1-D wave and for beam equations.
However, the regional result in 1-D case
 seems to be not extendable to these more general nonlinearities due
 to difficulties of employing the bound  \eqref{1Dbound} with $z$ replaced
 by $z_x$ or $z_t$.
\end{remark}

\begin{remark}
The result on the regional observability
can be easily extended to 1-D wave equations with variable coefficients
as considered in \cite{fridman2013observers}
$$\begin{array}{ll}
z_{tt}(x,t)={\partial \over
\partial x }[a(x) z_{x}(x,t)]+f(z(x,t), x,t),
 \\ t\geq t_0, \quad x\in [0,1],
\end{array}$$
where $a$ is a $C^1$ function with $a_x\leq 0$ and $a(1)>0$.
This can be done by  modifying Lyapunov and energy functions, where
the square of the partial derivative in $x$
should be multiplied by  $a(x)$.
Note that an extension of forward and backward observers to observability of
 1-D wave equations with non-Lipschitz coefficients
(as studied e.g. in \cite{castro02,Fanelli13}) seems to be problematic.
\end{remark}

\begin{example}
Consider \eqref{Wave_nd_z_nl_bound}
with $f=0.05z^2$. Here 
$|f_z|=|0.1z|\leq g_1$ if $|z|\leq 10 g_1=d$.
 Choose $g_1=0.1$, meaning that \eqref{locL} holds with $d=1$. Also here we use
  the sequence of forward and backward observers \eqref{zn} and \eqref{zbn}
with $k=1$.
  Verifying the feasibility of LMIs
  \eqref{psi1_Reg} and \eqref{phi_Reg} (subject to minimization of $\chi$ that enlarges the resulting $d_0$),
   we find that the system is exactly observable
  in time $T^*=3.78$, where $\delta=0.1$ and $\chi=0.1803$.
  This leads to the estimate  \eqref{region} with $d_0=0.2348$ for the region of
  exact observability, where the initial conditions of the system can be recovered
 uniquely from the measurements on the interval $[0,T]$ for all $T\in[3.78, 23.5]$.
 Note that the convergence of the iterative algorithm is faster for larger $T$ (in the sense that \eqref{lem_enThena} holds with a smaller $q$).
Increasing the nonlinearity twice to $f=0.1z^2$ and choosing $g_1=0.2$, we find $d=1$. The LMIs
  \eqref{psi1_Reg} and \eqref{phi_Reg} are feasible with $\delta=0.09, T^*=5.49$ and
  $\chi=0.2275$. We arrive at a smaller  $d_0=0.1867$, whereas $T\in [5.49, 15.4]$.

Simulations of the initial state recovery in the case of $f=0.1z^2$ and
 $   z_0(x) = z_1(x) = 0.2733 \cdot x(1- \frac{x}{2}) $, where
  $\int_0^1\left [{z_0}^2_x  + z_1^2 \right]dx =0.1867^2$,  show the convergence
 of the iterative algorithm on  the predicted observation
 interval $[0,5.49]$.
 Moreover, the algorithm converges on shorter observation intervals with $ T\geq 2.1 $
  that  illustrates the conservatism of the LMI conditions.
  See Figure \ref{fig1} for the case of $10$ forward and backward iterations with $T=1.8$ (no convergence) and $T=2.1$ (convergence).
The computation times for $10$ iterations for several values of $T$ are given in Table \ref{table:nonlin1}.

\begin{figure}[!ht]
\begin{center}
\includegraphics[height=4.5cm]{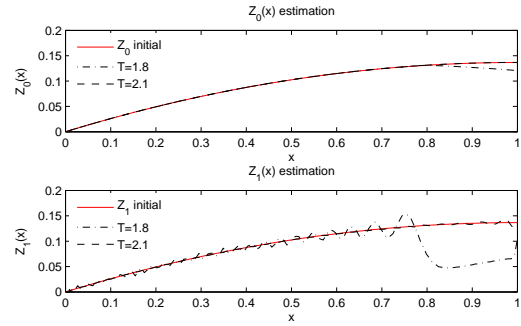}    
\caption{Initial condition recovery after $10$ iterations}  
\label{fig1}                                 
\end{center}                                 
\end{figure}

\begin{table}[ht]
\caption{Computation time for $10$ iterations}
\centering
\begin{tabular}{ c  c }
\hline \hline
T & Computation time (sec)  \\ [0.5ex]
\hline
    2.10  & 3.0469           \\
    3.00  & 3.6875           \\
    5.00  & 4.2813           \\
    10.00 & 5.9219           \\ [1ex]
\hline
\end{tabular}
\label{table:nonlin1}
\end{table}

\end{example}

\section{Conclusions}
The LMI approach to  observers and initial state recovering of
semilinear N-D wave equations on a hypercube has been presented. In the linear 2-D case our results lead to an upper bound on
the  exact observability time, which is close to the analytical value, but does not recover
it as it happened in 1-D case.
For 1-D systems with locally Lipschitz nonlinearities
we have found a (lower) bound on the region of initial values
that are uniquely recovered from the measurements on the finite interval.


\bibliographystyle{plain}        
\bibliography{Bibliography161013_02}           

\end{document}